\documentclass{lmcs}

\pdfoutput=1

\usepackage{lastpage}
\renewcommand{\dOi}{14(3:19)2018}
\lmcsheading{}{1--\pageref{LastPage}}{}{}%
{Feb.~24,~2018}{Sep.~14,~2018}{}

\usepackage{graphicx}
\usepackage{amssymb}
\usepackage{amsmath}
\usepackage{enumitem}
\usepackage{epstopdf}
\usepackage{amsthm}

\newcommand{\ov}[1]{\overline{#1}}
\newcommand{\bx}{\mathbf{x}}
\newcommand{\bby}{\mathbf{y}}
\newcommand{\W}{\mathrm{W}}
\newcommand{\R}{\mathbb{R}}

\newcommand{\eR}{\ov{\R}_+}

\newcommand{\myd}{\mathrm{d}}
\newcommand{\myi}{\mathrm{i}}
\newcommand{\myis}{\myi_*}
\newcommand{\myr}{\mathrm{r}}
\newcommand{\tuple}[1]{\langle #1 \rangle}
\newcommand{\supp}{\mathrm{supp}}
\newcommand{\tref}[2]{\hyperref[#2]{\ref*{#1}.\ref*{#2}}}

\title{Free Complete Wasserstein Algebras}
\author{Radu Mardare}
\address{Department of Computer Science, Aalborg University}
\email{mardare@cs.aau.dk}

\author{Prakash Panangaden}
\address{School of Computer Science, McGill University}
\email{prakash@cs.mcgill.ca}

\author{Gordon D.\ Plotkin}
\address{Laboratory for Foundations of Computer Science, School of Informatics, University of Edinburgh}
\email{gdp@inf.ed.ac.uk}

\begin{document}

%%%%%%%%%%%%%%%%
%%%%%%%%%%%%%%%%
%%%%%%%%%%%%%%%%

\begin{abstract}
We present an algebraic account of the Wasserstein distances  $W_p$ on complete metric spaces, for $p \geq 1$. This is part of a program of  a quantitative algebraic theory of effects in programming languages. 
In particular, we give axioms, parametric in $p$,  for algebras  over metric spaces equipped with probabilistic choice operations. The axioms say that the operations form a barycentric algebra and that the metric satisfies a property typical of the Wasserstein distance $W_p$. We  show that the free complete such algebra  over a complete metric space is that of the Radon probability  measures with finite moments of order $p$, equipped with the Wasserstein distance as metric   and with  the usual binary convex sums as operations.

\end{abstract}

\maketitle 
\section{Introduction}

The denotational semantics of probabilistic programs generally makes use of one or another monad of probability measures.  In~\cite{Law,Gir}, Lawvere and Giry  provided the first example of such a monad,  over the category of measurable spaces; Giry also gave a second example,  over Polish spaces. Later, in~\cite{JP,Jon}, Jones and Plotkin provided a monad of evaluations over directed complete posets. More recently, in~\cite{HKSY}, Heunen, Kammar, et al provided such a monad over quasi-Borel spaces.  

Metric spaces (and nonexpansive maps) form another natural category to consider. The question then arises as to what metric to take over probability measures, and the Kantorovich-Wasserstein, or, more generally, the Wasserstein  distances of order $p>1$ (generalising from $p = 1$), provide natural choices.  A first result of this kind was given in~\cite{VanB}, where it was shown that the Kantorovich distance  provides a monad over  the subcategory of 1-bounded complete separable metric spaces.

%{\tt check this}

One purpose of such monads is to make probabilistic choice operations available, and so an algebraic account of the monads in terms of these operations is of interest; we focus here on such  results in the metric context. Barycentric algebras axiomatise probabilistic choice.
% and midpoint algebras are two possibilities. Barycentric algebras 
They have binary convex combination operations $x +_r y$, for $r \in [0,1]$; one can think of $x +_r y$ as choosing to continue with $x$ with probability $r$ and with $y$ with probability $1-r$.
%Midpoint algebras have a midpoint operation $x \oplus y$; one can think of $x \oplus y$ as providing a fair choice between $x$ and $y$, and so being equivalent to $x +_{0.5} y$. In each case 
The operations are required to obey appropriate laws.

In this paper we show that the Wasserstein distance of order $p$ yields a monad $P_p$ on the category of complete metric spaces and nonexpansive maps and that this monad can be characterised as the free algebra monad for barycentric algebras over complete metric spaces where  
the barycentric operations are  required to be nonexpansive and are subject to an appropriate Wasserstein condition. The probability measures in $P_p(X)$ are required to be Radon (equivalently, tight) and to have finite $p$-moment. The important Kantorovich case of this result, where $p = 1$, was essentially already established in~\cite{FP}, as it is an immediate  consequence of Theorem 5.2.1 there.

In~\cite{MPP} a somewhat different direction was taken where monads on (extended) metric spaces are defined algebraically, using  a quantitative  analogue of equational logic where equations give an upper bound on the distance between two elements, and algebras are extended metric spaces equipped with nonexpansive operations.  One then seeks to characterise the action of the monad on as wide a variety of spaces as possible. 

This is part of a program begun in~\cite{MPP} to establish a quantitative algebraic theory of effects in programming languages. 
%
%In~\cite{BHM,BHMad}  it was essentially shown that the  Kantorovich monad of~\cite{VanB}, restricted to separable spaces, can be 
%characterised as the free algebra monad for midpoint algebras over metric spaces, where the midpoint operation is required to be 
%nonexpansive and subject to an appropriate `Kantorovich condition.' 
%
%We present an algebraic account of the Wasserstein distances $W_p$  on complete metric spaces. This is part of a program begun 
%in~\cite{MPP} to establish a quantitative algebraic theory of effects in programming languages. This program employs a quantitative 
%version of equational logic where equations give an upper bound on the distance between two elements, and algebras are extended 
%metric spaces equipped with nonexpansive operations. 
%
In particular, in~\cite{MPP} we gave axioms, parametric in $p$,  for algebras  equipped with probabilistic choice operations and showed the free complete such algebras over a complete 1-bounded separable metric space X are the probability  measures on $X$ with the Wasserstein distance, equipped with the usual binary convex sum operations. Comparing this result to the results in this paper and in~\cite{FP}, it is worth remarking that probability measures on 1-bounded spaces automatically have finite moments, and probability measures on separable complete spaces are automatically Radon.

%However, we only did this for the case where the metric space was 1-bounded and separable. 

%In this note we prove such a result for all complete metric spaces, but only for algebras over metric spaces, rather than extended ones. 
%We avoid the 1-bounded restriction by the standard idea of using probability measures with finite moments of order $p$; we avoid the %restriction to separable metric spaces by using Radon probability measures. 

In Section~\ref{Wasdis}, we discuss the Wasserstein distance, $W_p$ of order $p$ between probability measures on a metric space $X$.  In Theorem~\ref{CDgen} we show that the distance is a metric on the Radon probability measures with finite moments of order $p$. To do so, we make use of a result in~\cite{CD} that the triangle inequality holds for all probability measures if $X$ is separable. Here, and throughout the paper, we reduce the general case to the separable case using a lemma, due to Basso, which states that the support of any probability measure on a metric space is separable.

Next, in Theorem~\ref{Bolgen}
we show that if $X$ is complete then the  Wasserstein metric of order $p$ on the Radon probability measures with finite moments of order $p$ is also complete and is generated by the probability measures with finite support. To do this, we make use of the well-known result that if $X$ is complete and separable  then the  Wasserstein metric of order $p$ on all probability measures over $X$ with finite  moments of order $p$ is also complete and separable, being generated by the rational measures with finite support in a countable basis of  $X$ (see, e.g.,~\cite[Theorem 6.18]{Vil08} and the bibliographic discussion there). The section concludes with a side-excursion, a discussion of weak convergence: Theorem~\ref{mettop} generalises the characterisation of the topology induced by the Wasserstein metrics in terms of weak convergence given by~\cite[Theorem 6.9]{Vil08} from complete separable spaces to all complete spaces.

 In Section~\ref{Wasalg} we discuss  the algebraic aspect of these spaces. In particular, in Theorem~\ref{wholepoint} we show that the Radon probability measures with finite  moments of order $p$ on a complete metric space $X$,  equipped with the $W_p$ metric and  binary convex sums, form the free complete Wasserstein  algebra of order $p$ over $X$. 
 %(The important case $p = 1$ of the Kantorovich monad was essentially already proved in~\cite{FP}, as it is an immediate  consequence 
 %of Theorem 5.2.1 there.) 
  To do this, we first characterise the free Wasserstein  algebra of order $p$ over a metric space; the characterisation follows straightforwardly from the well-known characterisation of the free barycentric algebra over a set as the natural such algebra on the finite probability measures on the set. Having done so, we pass to the case of complete metric spaces using a general theorem on completions of metric algebras, combined with  Theorem~\ref{Bolgen}.

Finally, we discuss some loose ends. First, there is a certain disconnect between our work and that of~\cite{MPP} in that here we use standard metric spaces, whereas  there the more general framework of extended metric spaces is employed where distances can be infinite.  We sketch how to bridge this disconnect at the end of Section~\ref{Wasalg}. Second,  there are other natural algebraic approaches to  probability. Convex  spaces are an algebraic formulation of finite convex combinations $\sum_{i = 1}^ n r_ix_i$. Midpoint algebras have a midpoint operation $x \oplus y$: one can think of $x \oplus y$ as providing a fair choice between $x$ and $y$, and so being equivalent to $x +_{0.5} y$. We sketch how our algebraic characterisation of the Wasserstein monads on complete metric spaces can be rephrased in terms of either of these alternative approaches.

We affectionately dedicate this paper to Furio Honsell on the occasion of his 60th birthday.
%%%%%%%%%%%%%%%%
%%%%%%%%%%%%%%%%
%%%%%%%%%%%%%%%%

\section{The Wasserstein distance} \label{Wasdis}

%{\tt add references for various assertions}

We begin with some technical preliminaries on Radon probability measures and couplings and their support. For general background on probability measures on topological and metric spaces see~\cite{Fre,Par}. By probability measure we mean a Borel probability measure. Given such a probability measure $\mu$ on a Hausdorff space $X$, we say that a Borel set $B$ is \emph{compact inner regular (for $\mu$)} if:
\[\mu(B) = \mathrm{sup} \{\mu(C) \mid C \mbox{ compact}, C \subseteq B \}\]
 Then $\mu$ is \emph{Radon} if all Borel sets are compact inner regular for it, and \emph{tight} if $X$ is compact inner regular for it (equivalently,  if for any Borel set $B$ and $\varepsilon > 0$ there is a compact set $C$ such that $\mu(B\backslash C) < \varepsilon$). 
  Every Radon measure is tight, every tight probability  measure on a metric space is Radon, and every probability  measure on a separable complete metric space is tight.

%{\tt insert reference}

The \emph{support} of a probability measure $\mu$ on a topological space $X$ is
\[\supp(\mu) =_{\mathbf{def}} \{x \in X \mid \mu(U) > 0 \mbox{ for all open $U$ containing $x$}  \}\]
Note that the support is always a closed set. If $\mu$ is Radon then $\supp(\mu)$ has measure 1.
 If the support of $\mu$ is finite then $\mu$ is Radon and can be written uniquely, up to order, as a finite convex sum of Dirac measures, viz:
\[\mu = \sum_{s \in \supp(\mu)} \mu(s)\delta(s)\]
(writing $\mu(s)$  instead of $\mu(\{s\})$). We say that $\mu$ is \emph{rational} if all the $\mu(s)$  are.

The following very useful lemma is due to Basso~\cite{Bas}. It enables us, as he did, to establish  results about probability measures on  metric spaces by applying  results about probability measures on separable metric spaces to their supports. 
\begin{lem} \label{basic} 
Every  probability measure on a metric space has separable support.
 \end{lem} 
\begin{proof}
Let $\mu$ be a probability measure on a metric space. For every $\varepsilon > 0$ let $C_{\varepsilon}$ be a maximal set of points in $\supp(\mu)$ with the property that all points are at distance $\geq \varepsilon$ apart. As any two open balls with centre in $C_{\varepsilon}$  and radius $\frac{1}{2}\varepsilon$ are disjoint and each such ball has $\mu$-measure $>0$,  $C_{\varepsilon}$  is countable (were it uncountable, we would have an uncountable collection of reals with any denumerable subset summable, and no such collection exists). By the maximality of  $C_{\varepsilon}$, any point in $\supp(\mu)$ is at distance $< \varepsilon$ from some point in $C_{\varepsilon}$. It follows that the countable set $\bigcup_{n \geq 0} C_{2^{- n}}$ is dense in $\supp(\mu)$.
\end{proof}

A \emph{coupling} $\gamma$  between  two probability measures $\mu$ and $\nu$ on a topological space $X$ is a  probability measure  on  $X^2$ whose left and right marginals (= pushforwards along the projections) are, respectively, $\mu$ and $\nu$. We gather some facts about such  couplings:
\begin{lem} \label{coup} Let $\gamma$  be a coupling between  two  probability measures $\mu$ and $\nu$ on a topological space $X$. Then:
\[\supp(\gamma) \subseteq \supp(\mu) \times \supp(\nu)\]
Further, if $\mu$ and $\nu$ are tight, so is $\gamma$. Moreover, in the case that $X$ is a metric space, if $\mu$ and $\nu$ are Radon,  so is  $\gamma$. 
%\begin{enumerate}
%\item \label{coup1} \[\supp(\gamma) \subseteq \supp(\mu) \times \supp(\nu)\]
%\item \label{coup2}
%\end{enumerate}
\end{lem}
\begin{proof}
For the first part, suppose that $(x,y)\in \supp(\gamma)$ and let $U$ be an open neighbourhood of $x$. Then $U \times Y$ is an open neighbourhood of $(x,y)$ and so $\mu(U) = \gamma(U \times Y) > 0$. So $x \in \supp(\mu)$. Similarly $y \in \supp(\nu)$.

For the second part, choose $\varepsilon > 0$. As $\mu$ and $\nu$ are tight, $X$ and $Y$ are compact inner regular for them. So there are compact sets $C \subseteq X$ and $D \subseteq Y$ such that $\mu(X\backslash C) < \frac{1}{2}\varepsilon$ and $\nu(Y\backslash D) < \frac{1}{2}\varepsilon$.  Then we have:
\[\begin{array}{lcl}\gamma((X \times Y) \backslash (C \times D) ) 
          & = &  \gamma((X \times Y) \backslash ((C \times Y) \cap (X \times D)))\\
         & = &  \gamma(((X \times Y) \backslash (C \times Y) )\cup ((X \times Y) \backslash (X \times D)))\\
         & = &  \gamma(((X \backslash C) \times Y)\cup (X \times (Y \backslash D) ))\\
         & \leq &  \gamma((X \backslash C) \times Y) + \gamma(X \times (Y \backslash D) )\\
         & = &  \mu(X \backslash C)  + \nu(Y \backslash D)\\
         & < & \varepsilon
\end{array}\]
So, as $\varepsilon$ was arbitrary, $X \times Y$ is compact inner regular for $\gamma$, as required. The last part is immediate as all tight probability measures on metric spaces are Radon.
\end{proof}
%
%It follows from the lemma that  any coupling between Radon probability measures on metric spaces is itself Radon.

%%%%%%%%%%%%%%%%
%%%%%%%%%%%%%%%%
%%%%%%%%%%%%%%%%
%
%
%\section{TheWasserstein metric}
We now turn to the Wasserstein distance. 
 A  probability measure $\mu$  on a metric space $X$  is said to have \emph{finite   moment of order $p$}, where $p \geq1$, if, for some (equivalently 
 %
% {\tt insert reference}
 %
 all) $x_0 \in X$, the integral
\[\int d(x_0, -)^p\myd\mu\]
is finite
To see that the existence of the $p$-th moment does not depend on the choice of $x_0$, recollect the inequality $(a + b)^p \leq 2^{p-1}(a^p + b^p)$ for $a,b \in \R_+$ and $p \geq 1$. Then,  for any $x,y,z \in X$, we have:
  \begin{align} d(x,y)^p &\leq 2^{p-1}(d(x,z)^p + d(z,y)^p)
  \end{align}
and the conclusion follows, taking $z = x_0$ and $y$ to be any other choice. Note that  finite probability measures with finite support have finite moments of all orders.

\begin{exa} \label{example1}
One can obtain examples of measures with countable support, and with or without finite moments, by using the fact  that the Dirichlet series $\zeta(s) = \sum_{n \geq 1}\frac{1}{n^s}$ converges for reals $s > 1$ and diverges for $s = 1$. Taking the  natural numbers as a metric space with the usual Euclidean metric $d(m,n) = |m - n|$, one sees that, for $q \geq 1$, the discrete probability measure $D_q$ where
$D_q(n) = \zeta(q+1)^{-1}\frac{1}{n^{q +1}}$ for $n \geq 1$ and $D_q(0) = 0$,  has finite $p$-moment if, and only if,  $p < q$. 
\end{exa}

The Wasserstein distance of order $p$,  $\W_p$ is  defined between  probability measures with finite  moments of order $p$ by:
\[\W_p(\mu,\nu) = \left  ( \inf_{\gamma} \int d^p_X \myd\gamma \right )^{1/p} \]
where $\gamma$ runs over the %Radon 
couplings between $\mu$ and $\nu$. 
%(the integrals are finite as $\mu$ and $\nu$ have finite  moments of order $p$).  
To see that it is well-defined for probability measures with finite $p$-moment, one again invokes (1), this time to get an integrable upper bound on $d^p_X$ as follows:
\[\int d^p_X \myd\gamma \leq 2^{p-1}(\int d(- ,x_0)^p \myd \gamma  + \int d(x_0,-)^p\myd \gamma) = 2^{p-1}(\int d(- ,x_0)^p \myd \mu  + \int d(x_0,-)^p\myd \nu)\]
The Wasserstein distance in monotonic in $p$, that is, $\W_p(\mu,\nu) \leq \W_q(\mu,\nu)$ for $1 \leq p \leq q$ and $\mu,\nu$ with finite moments of order $q$. This is an immediate consequence of the  inequality $(\int f^p \myd \mu)^{1/p} \leq (\int f^q   \myd\mu)^{1/q}$, which holds for any probability measure $\mu$ and any $f \geq 0$ with $f^q$ integrable w.r.t.\  $\mu$ (this inequality is itself a straightforward consequence of H\"{o}lder's inequality). 
 %It is shown in~\cite{CD} that this defines a metric when $X$ is separable.
 %%
 % If, further, $X$ is complete it is shown in~\cite{Bol} 
%%{\tt any other references?}\\
%that this defines a complete separable metric space generated by the finitely supported probability 
%measures with rational coefficients over a countable basis for $X$. 
%We use these results to obtain results for spaces of Radon probability measures over a metric 
%space without any assumption of separability. The technique is the same as in Basso~\cite{Bas}: 
%one uses the fact that Radon probability measures have countable support with measure 1 to 
%reduce questions about metric spaces to corresponding questions about separable subspaces.

We need two lemmas relating probability measures on a  metric space with probability measures on a closed subset of the space.
Let $C$ be a closed subset of a metric space $X$. 
%Pushforward  along the inclusion $\iota: C \rightarrow X$ defines an injection  $\myis_{C,X}: P_p(C) 
%\rightarrow  P_p(X)$. 
We write $\myis(\mu)$ for the pushforward of a finite measure $\mu$ on $C$ along the inclusion $\myi: C \rightarrow X$ of $C$ in $X$, so $\myis(\mu)(B) = \mu(B \cap C)
$; and we write  $\myr(\nu)$ for the restriction of a finite measure $\nu$ on $X$ to a finite measure on $C$, so $\myr(\nu)(B) = \nu(B)$. Note that 
$\myr(\myis(\mu)) = \mu$.

%Let $C$ be a closed subset of a metric space $X$. 
%%Pushforward  along the inclusion $\iota: C \rightarrow X$ defines an injection  $\myis_{C,X}: P_p(C) 
%%\rightarrow  P_p(X)$. 
%We write $\myis_{C,X}(\mu)$ for the pushforward of a finite measure $\mu$ on $C$ along the inclusion of $C$ 
%in $X$, so $\myis_{C,X}(\mu)(B) = \mu(B \cap C)
%$; we write  $\myr_{C,X}(\nu)$ for the restriction of a finite measure $\nu$ on $X$ to a finite measure on $C$, 
%so $\myr_{C,X}(\nu)(B) = \nu(B)$. Note that 
%$\myr_{X,C}(\myis_{X,C}(\mu)) = \mu$. We omit the 
%suffices when they can be understood from the context.

\begin{lem} \label{rel}
Let $C$ be a closed subset of a metric space $X$. 
\begin{enumerate} 
\item \label{rel1} If $\mu$ is a Radon probability measure on $C$, then $\myis(\mu)$ is a    Radon probability measure on $X$ with the same support as $\mu$.
Further,  $\myis(\mu)$  has finite  moment of order $p$ if $\mu$ does.

\item  \label{rel2} If $\nu$ is a Radon probability measure on $X$ with support included in $C$ then  $\myr(\nu)$ is a Radon probability measure on $C$, and we have:
\[\nu = \myis(\myr(\nu))\]
If, further,  $\nu$   has finite  moment of order $p$, so does $\myr(\nu)$.
\end{enumerate}
\end{lem}
\begin{proof}
\begin{enumerate}
\item Let  $\mu$ be a Radon probability measure on $C$, and note that $C$ must then be non-empty. It is straightforward to check that 
$\myis(\mu)$ is a    Radon probability measure on $X$ with the same support as $\mu$. Next,
choosing $x_0 \in C$, as 
\[\int d_X(x_0,-)^p \myd (\myis\mu) = \int d_C(x_0,-)^p \myd \mu\]  
we see that $\myis(\mu)$  has finite  moment of order $p$ if $\mu$ does.

\item Let $\nu$ be a Radon probability measure on $X$ with support included in $C$. %Then $C$ is nonempty. Further, 
It is straightforward to check that  $\myr(\nu)$ is a Radon probability measure on $C$. Regarding the equality, for any Borel set $B$ of $X$ we have:
\[\myis(\myr(\nu))(B) = \myr(\nu)(B \cap C) = \nu(B \cap C) =  \nu(B) \]
with the last equality holding as $\nu$ is Radon and the support of $\nu$ is included in $C$. Finally, $\myr(\nu)$ has finite  moment of order $p$ if $\nu$ does,  as, choosing $x_0 \in C$, we have: 
\begin{align*}
\int d_C(x_0,-)^p \myd (\myr(\nu)) =  \int d_X(x_0,-)^p \myd (\myis(\myr (\nu))) = \int d_X(x_0,-)^p \myd \nu \tag*{\qedhere}
\end{align*}
\end{enumerate}

\end{proof}

\begin{lem} \label{red} Let $C$ be a closed subset of a metric space $X$. Then
\begin{enumerate}
\item \label{red1}  For any Radon probability measures $\mu, \nu$ on $X$ with finite   moments of order $p$  we have:
\[\W_p(\mu,\nu) = \W_p(\myis(\mu),\myis(\nu))\]
\item  \label{red2}  For any   Radon probability measures $\mu, \nu$ on $X$ with finite   moments of order $p$   whose support is included in $C$,  we have:
\[\W_p(\mu,\nu) = \W_p(\myr(\mu),\myr (\nu))\]

\end{enumerate}
\end{lem}
\begin{proof}
\begin{enumerate}
\item Let $\gamma$ be a %Radon 
coupling between $\mu$ and $\nu$. 
%Then, by Lemma~\tref{rel}{rel1}, $\myis(\gamma)$ is a Radon probability measure on 
%$X^2$.  It is also 
Then $(\myi \times \myi)_*(\gamma)$ is a coupling between $\myis(\mu)$ and $\myis(\nu)$, as
$(\pi_0)_*((\myi \times \myi)_*(\gamma)) = \myis((\pi_0)_*(\gamma)) = \myis(\mu)$, and similarly for $\nu$.
%
%\[\begin{array}{lcl}
%  \myis(\gamma)(B \times X)  & = & \gamma((B \times X) \cap (C \times  C))\\
%                                               & = &   \gamma((B \cap C) \times C) \\ 
%                                               & = &  \mu(B \cap C) \\
%                                               & = & \myis(\mu)(B)
%\end{array}\]
%
%and similarly for $\nu$. 
We also have:
\[\int d^p_X \myd ((\myi \times \myi)_*\gamma) = \int d^p_C \myd \gamma\]
As $\gamma$ was chosen arbitrarily, we therefore have:
\[\W_p(\mu,\nu) \geq \W_p(\myis(\mu),\myis(\nu))\]

For the reverse inequality, let $\gamma$ be a %Radon 
coupling between $\myis(\mu)$ and $\myis(\nu)$. 
By Lemma~\tref{rel}{rel1}, $\myis(\mu)$ and $\myis(\nu)$ are Radon. So, by Lemma~\ref{coup}, $\gamma$ is also Radon. Also,  $\supp(\gamma) \subseteq C \times C$, since $\supp(\myis(\mu)) \subseteq C$ and $\supp(\myis(\nu)) \subseteq C$. 
So, by Lemma~\tref{rel}{rel2}, $\myr(\gamma)$ is a Radon probability measure on $C^2$ and $\gamma = (\myi \times \myi)_*(\myr(\gamma))$. 

Further, $\myr(\gamma)$ is a coupling between $\mu$ and $\nu$, for:
\[\begin{array}{lcll}\myr(\gamma)(B \times C) & = &   \gamma(B \times C)\\ 
                                                                      & = &  \gamma( (B \times X) \cap (C \times C) )\\
                                                                        & = &  \gamma(B \times X) \quad (\mbox{as $\supp(\gamma) \subseteq C \times C$})\\
                                                                          & = & \mu(B)\\
                                                                           & = &  \myr(\mu)(B) 
 \end{array}
\]
 and similarly for $\nu$.
We then have:

 \[\int d^p_X \myd \gamma = \int d^p_X \myd ((\myi \times \myi)_*(\myr\gamma)) =  \int d^p_C \myd (\myr\gamma)\]
 As $\gamma$ was chosen arbitrarily, we therefore have, as required:
\[\W_p(\myis(\mu),\myis(\nu)) \geq \W_p(\mu,\nu)\]
%
%Combining the two inequalities, we have:
%%
%\[\W_p(\mu,\nu) = \W_p(\myis(\mu),\myis(\nu))\]
%%
%as required.
\item
Using part (1), we have:
\begin{align*}
\W_p(\mu,\nu) = \W_p(\myis(\myr(\mu)),\myis(\myr (\nu))) = \W_p(\myr(\mu),\myr (\nu)) \tag*{\qedhere}
\end{align*}
\end{enumerate}

\end{proof}

With these technical lemmas established, we can now prove two  theorems on spaces of Radon probability measures. For any metric space $X$ and $p \geq 1$, define $P_p(X)$ to be the set  of Radon probability measures on $X$ with finite   moments of order $p$,  equipped  with the $W_p$ distance. For the first theorem we use the result in~\cite{CD}  that the triangle inequality holds for all probability measures if $X$ is separable.
\begin{thm} \label{CDgen} Let $X$ be  a metric space. Then $P_p(X)$ is a  metric space. 
\end{thm}
\begin{proof}
First $\W_p(\mu,\mu) = 0$ for any $\mu \in P_p(X)$ as $\Delta_*\mu$, the pushforward of $\mu$ along the diagonal $\Delta: X \rightarrow X \times X$,  is a coupling between $\mu$ and itself. 
For the converse, choose $\mu , \nu \in P_p(X)$, and  suppose  $\W_p(\mu,\nu) = 0$. 
Then  $\W_1(\mu,\nu) = 0$, as $W_q$ is monotonic in $q$.
%
%\[\W_1(\mu,\nu)^{1/p} = (\inf_{\gamma} \int d_X \myd \gamma)^{1/p} \leq  \inf_{\gamma} (\int d_X \myd \gamma)^{1/p} \leq  \inf_{\gamma} (\int d_X^p \myd \gamma)^{1/p} = \W_p(\mu,\nu)\]
%
It follows that $\mu = \nu$ as $W_1$ is  a metric on $P_1(X)$ (see, e.g.,~\cite{Edw,Bas}).

%For the converse, suppose $\W_p(\mu,\nu) = 0$ for some $\mu , \nu \in P_p(X)$.  Let $C$ be the 
%closed set $\supp(\mu) \cup   \supp(\nu)$. By Lemma~\tref{rel}{rel2}, $\myr (\mu), \myr (\nu) \in   
%P_p(C)$. By Lemma~\tref{red}{red2}, we have $\W_p(\myr (\mu),\myr (\nu)) = 0$. As $C$ is 
%separable, by Lemma~\ref{basic}, by~\cite{CD}, $\W_p$ is a metric on $C$. So we have  $\myr 
%(\mu) = \myr (\nu)$, and so $\mu = \nu$ by Lemma~\tref{rel}{rel2}.

Symmetry is evident. To show the triangle inequality, suppose that $\mu,\nu, \omega \in P_p(X)$. Let $C$ be the closed set $\supp(\mu) \cup \supp(\nu) \cup \supp(\omega)$, separable by Lemma~\ref{basic}. 
Then $\myr (\mu)$, $\myr (\nu)$ and $\myr (\omega)$ are probability measures on the separable space $C$, and so by~\cite{CD}, we have $\W_p(\myr (\mu),\myr (\omega)) \leq \W_p(\myr (\mu),\myr (\nu)) + \W_p(\myr (\nu), \myr (\omega))$. Then, by Lemmas~\tref{red}{red2} and~\tref{rel}{rel2}, we see that  $\W_p(\mu,\omega) \leq \W_p(\nu,\omega) + \W_p(\nu,\omega)$, as required. 
\end{proof}

For the second theorem we use the result that the metric space of all probability measures with finite  moments of order $p$ on a complete and separable space %$X$ 
is also complete and separable, being generated by the rational measures with finite support in a countable basis of  the space. %$X$.

%We next generalise~\cite{Bol} to all complete metric spaces.
\begin{thm} \label{Bolgen} Let $X$ be  a complete metric space. Then $P_p(X)$ is a  complete metric space generated by the finitely supported probability measures on $X$.  
%Further,  equipped with the standard convex combination operations $\mu +_r \nu$, for $r\in [0,1]$, $P_p(X)$ 
%forms a complete Wasserstein algebra of order $p$. 
\end{thm} 
\begin{proof} To show that $ P_p(X)$ is complete, let $\langle \mu_i\rangle_i$ be a Cauchy sequence in $P_p(X)$. Let $C$ be the closure in $P_p(X)$ of $\bigcup_i \supp(\mu_i)$. By Lemma~\ref{basic}, each set $\supp(\mu_i)$ is separable, and so $C$ is itself separable, being the closure of a countable union of separable sets. Further, $C$ is complete as it is a closed subset of a complete space.  So $P_p(C)$ is complete. Applying Lemmas~\tref{rel}{rel2} and~\tref{red}{red2}, we see that $\langle \myr (\mu_i) \rangle_i$ is a Cauchy sequence in $P_p(C)$. Let $\mu$ be its limit there.
 %(by~\cite{Bol}, $P_p(C)$ is complete as $C$ is complete and separable). 
 As  Lemma~\tref{red}{red1} shows that $\myis$ is an isometric embedding, we see that  $\langle \myis(\myr (\mu_i)) \rangle_i$ is a Cauchy sequence in $P_p(C)$ with limit $\myis (\mu)$. But, by Lemma~\tref{rel}{rel2}, $\langle \myis(\myr (\mu_i)) \rangle_i$ is %exactly 
 $\langle \mu_i\rangle_i$ .

To see that the finitely supported probability measures are dense in $P_p(X)$, choose $\mu \in P_p(X)$ and $\varepsilon > 0$. Then, taking $C$ to be the separable closed set  $\supp(\mu)$,  we  have $\myr (\mu) \in P_p(C)$. Then, as $P_p(C)$ is generated by the rational measures with finite support,  there is a finitely supported probability measure $\alpha \in P_p(C)$ at distance $\leq \varepsilon$ from $\myr (\mu)$, and so we see that $\myis( \alpha)$ is at distance $\leq \varepsilon$ from $\mu$. Finally, $\myis( \alpha)$ is finitely supported as, by Lemma~\tref{rel}{rel1}, it has the same support as  $\alpha$.
\end{proof}

It is interesting to consider how the different Wasserstein metrics generate the various $P_p(X)$  starting from the same basis, i.e., the probability measures with finite support. 
As $W_p$ is monotonic in $p$ when $p \leq q$, any sequence Cauchy in the $W_q$ metric is also Cauchy in the $W_p$ metric. So the difference must be that sequences of probability measures with finite support can be Cauchy in the $W_p$ metric, but not in the $W_q$ metric when $p < q$. Examples of this phenomenon can be found by building on Example~\ref{example1}:

\begin{exa} Set
\[D_{q,m}(n) = \left \{\begin{array}{ll}    0  & (n = 0)\\
                                                            \zeta(q+1)^{-1}\frac{1}{n^{q +1}} & (1 \leq n \leq m)\\ 
                                                            1 -  \zeta(q+1)^{-1}\sum^m_{i = 1}\frac{1}{n^{q +1}} & (n = m + 1)\\
                                                            0 & (n >m  + 1)
                                \end{array}\right .\]
Note that $D_{q,m}$ and $D_q$ agree for $1 \leq n \leq m$, but the rest of the mass of $D_q$ is concentrated on $D_{q,m}$ at $m + 1$. Then the sequence $D_{q,m}$ converges to $D_q$ in the $W_p$ metric when $p < q$, but is not Cauchy w.r.t.\ $W_p$ when $p = q$.

\end{exa}

We pause our development to examine the topologies induced by the Wasserstein metrics.
 In the case of separable such spaces it is known that these metrics topologise a suitable notion of weak convergence, see~~\cite[Theorem 7.12]{Vil03}.
It turns out that we can again use our techniques and generalise these results to all complete metric spaces.

Given a metric space $X$ and probability measures $\mu_i$ ($i \geq 0$) and $\mu$  on $X$, we say that the $\mu_i $ \emph{converge weakly} to $\mu$, and write $\mu_i \longrightarrow \mu$ if, for all continuous bounded $f:X \rightarrow \R$, we have $\int f \myd \mu_i \longrightarrow \int f \myd \mu$; an equivalent formulation is that $\underline{\lim}_i \, \mu_i(U) \geq  \mu(U)$ for every open set $U$ (see ~\cite[Theorem 6]{Par}).

 \begin{lem} \label{weak} Let $C \subseteq X$ be a closed subset of a metric space and let $\mu_i$ ($i \geq 0$) and $\mu$ be probability measures on $C$. Then we have:
 \[ \mu_i \longrightarrow \mu \quad \iff  \quad \myis(\mu_i) \longrightarrow \myis(\mu)\]
 where $\iota:C \rightarrow X$ is the inclusion map.
 \end{lem}
 \begin{proof} We use the equivalent formulation of weak convergence in terms of measures of open sets. In one direction, suppose that  
 $\mu_i \longrightarrow \mu $. Then 
 $\myis(\mu_i) \longrightarrow \myis(\mu)$, since for every open subset $U \subseteq X$, and any probability measure $\nu$ on $C$, we have $\myis(\nu)(U)  = \nu(U \cap X)$.
 In the other direction, suppose that  $\myis(\mu_i) \longrightarrow \myis(\mu)$. Then  $\mu_i \longrightarrow \mu $ since for any open set $U$ of $C$ there is an open set $V$ of  $X$ such that $U = V \cap C$ and then, for any probability measure $\nu$ on $C$, we have $\nu(U) = \myis(\nu)(V)$.
   \end{proof}

 \begin{lem} \label{weakp} Let $C$ be a separable closed subset of a complete  metric space $X$, and let $\myi  :   C \rightarrow X$ be the inclusion map. Then, for any probability measures $\alpha_i$ ($i \geq 0$) and $\alpha$ on $C$, and any $x_0 \in C$, the following are equivalent, and hold independently of the choice of $x_0 \in C$: 
  \begin{enumerate}
 \item $\myi_*(\alpha_i) \longrightarrow \myi_*(\alpha)$ and   $\int d_X(x_0, -)^p \myd \myi_*(\alpha_i) \longrightarrow  \int  d_X(x_0, -)^p \myd \myi_*(\alpha)$.
  \item $\alpha_i \longrightarrow \alpha$ and   $\int d_C(x_0, -)^p \myd \alpha_i \longrightarrow  \int  d_C(x_0, -)^p \myd\alpha$.
 \item   For all continuous functions  $ f  :  C   \rightarrow \R$ such that $|f(x)| \leq   c (1 + d_C(x_0, x)^p)$ for all $x\in C$, for some $c \in \R$, 
one has $\int f \myd\alpha_i \longrightarrow  f \myd\alpha$.
 \item  For all continuous functions  $f  :   X  \rightarrow \R$ such that $|f(x)| \leq  c (1 + d_X(x_0, x)^p)$ for all $x\in X$, for some $c \in \R$, 
one has $\int f \myd\myi_* (\alpha_i) \longrightarrow  f \myd\myi_* (\alpha)$.
 \end{enumerate}
 \end{lem}
 \begin{proof}
 As $C$ is separable and complete, by~\cite[Theorem 7.12]{Vil03}, (2) and (3) are equivalent. and hold independently of the choice of $x_0 \in C$. As, further, (4) trivially implies (1), it suffices to prove that (1) implies (2) and (3) implies (4), for any $x_0 \in C$.
 The first of these implications follows using Lemma~\ref{weak} and the fact that 
 $\int d_C(x_0, -)^p \myd \alpha_i = \int d_X(x_0, -)^p \myd \myi_*(\alpha_i) $, 
 %  $\int d_X(x_0, -)^p \myd \myi_*(\alpha_i) = \int d_C(x_0, -)^p \myd \alpha_i$, 
 for all $i \geq 0$, and similarly for $\alpha$.
 For the second of these implications, one again uses Lemma~\ref{weak} and notes that if  $f  :  X   \rightarrow \R$ is a continuous function such that $|f(x)| \leq   c (1 + d_X(x_0, x)^p)$ for all $x \in X$, for some $c \in \R$, then $f\myi  :   C   \rightarrow \R$ is a continuous function such that $|(f\myi)(x)| \leq   c (1 + d_C(x_0, x)^p)$,  for all $X \in C$ and that 
 $\int f \myd \myis(\alpha_i) = \int (f\myi) \myd \alpha_i$, and similarly for $\alpha$.
\end{proof}

\begin{thm} \label{weak-conv}  Let $X$ be a complete metric space and choose $\mu_i$ ($i \geq 0$) and $\mu$  in $P_p(X)$. Then the following two conditions hold independently of the choice of  $x_0 \in X$ and are equivalent:
 
 \begin{enumerate}

 \item $\mu_i \longrightarrow \mu$ and   $\int d(x_0, -)^p \myd \mu_i \longrightarrow  \int  d(x_0, -)^p \myd\mu$,
 
 \item  For all continuous functions  $f: X  \rightarrow \R$ such that $|f(x)| \leq  c (1 + d(x_0, x)^p)$ for all $x \in X$, for some $c \in \R$, 
one has $\int f \myd\mu_i \longrightarrow  f \myd\mu$.
 \end{enumerate}
 
 \end{thm}
 \begin{proof} First, choose $a,b \in X$. Next, choose a separable closed set $C \subseteq X$ containing $a$ and $b$ and including  the supports of  the $\mu_i$ ($i \geq 0$) and $\mu$.
 Then, taking   $\alpha_i = r(\mu_i)$ ($i \geq 0$) and $\alpha = r(\mu)$ in Lemma~\ref{weakp}, we see that (1) and (2) hold independently of the choice of $x_0 \in C$  (e.g., whether choosing $a$ or $b$) and are equivalent for any such choice as   $\mu_i = \myi_*(\alpha_i)$, for all $i \geq 0$, and $\mu = \myi_*(\alpha)$.
 As $a$ and $b$   were chosen arbitrarily from $X$,  the conclusion follows.
  \end{proof}

 If either of the two equivalent conditions of Theorem~\ref{weak-conv} hold  for $\mu_i$ ($i \geq 0$) and $\mu$  in $P_p(X)$ and any choice of $x_0 \in X$, where $X$ is a complete metric space, we say that the $\mu_i$ \emph{converge $p$-weakly} to $\mu$, and write $\mu_i \longrightarrow_p \mu$.

 %{\tt say something about Villani's other conditions}

 %
 \begin{thm} \label{mettop} For any complete metric space $X$,  the Wasserstein distance $W_p$ metricises $p$-weak convergence in $P_p(X)$.
\end{thm} 
\begin{proof}  Choose $\mu_i$ ($i \geq 0$) and $\mu$ in $P_p(X)$ to show that $\mu_i \longrightarrow_p \mu$ iff $\mu_i$ converges to $\mu$ w.r.t.\ the Wasserstein distance $W_p$ in $X$.
To this end, let $C \subseteq X$ be a closed separable set containing the supports of the $\mu_i$ and $\mu$ (and so necessarily nonempty), and set $\alpha_i = r(\mu_i)$ and $\alpha = r(\mu)$ and let $x_0$ be an element of $C$.  By~\cite[Theorem 6.9]{Vil08} we have that $\alpha_i \longrightarrow_p \alpha$ iff $\alpha_i$ converges to $\alpha$ w.r.t.\ the Wasserstein distance $W_p$ in $C$.
So it suffices to show (i) that $\mu_i \longrightarrow_p \mu$ iff $\alpha_i \longrightarrow_p \alpha$, and (ii) that $\mu_i$ converges to $\mu$ w.r.t.\ the Wasserstein distance $W_p$ in $X$ iff $\alpha_i$ converges to $\alpha$ w.r.t.\ the Wasserstein distance $W_p$ in $C$.

That (i) holds follows from  the definition of $p$-weak convergence and the equivalence between parts (1) and (2) of Lemma~\ref{weakp} applied to the $\alpha_i$, $\alpha$, and $x_0$. That (ii) holds follows from the fact that, by Lemma~\tref{red}{red1}, the inclusion of $C$ in $X$ is an isometric embedding.
\end{proof}

%%%%%%%%%%%%%%%%
%%%%%%%%%%%%%%%%
%%%%%%%%%%%%%%%%

\section{Wasserstein algebras} \label{Wasalg}
%%%%%%%%%%%%%%%%
%%%%%%%%%%%%%%%%
%%%%%%%%%%%%%%%%

\newcommand{\myomit}[1]{}

We begin with an account of barycentric algebras. We then move on to quantitative algebras, by which we mean algebras over metric spaces. After some general considerations on these we move on to Wasserstein algebras of order $p$, and our main theorem, characterising the Wasserstein monads $P_p$.

 A  \emph{barycentric algebra} (or  \emph{abstract convex set}) is a
  set $X$ equipped with binary 
  operations $+_r$ for every real number $r \in [0,1]$
  such that the following equational laws
hold:
\begin{align} 
x+_1y &= x \label{B1} \tag{B1}\\
x+_rx& = x \label{B2} \tag{B2} \\
x+_ry& = y+_{1-r}x\label{SC} \tag{SC}\\  
(x+_py)+_r z &= x+_{pr}(y+_{\frac{r-pr}{1-pr}}z)\ \ \  \mbox{ provided }
r<1, p<1\label{SA} \tag{SA}
\end{align}
SC stands for \emph{skew commutativity} and SA for \emph{skew associativity}.
Homomorphisms of barycentric algebras are termed \emph{affine}.

One can inductively define finite convex sums in a barycentric algebra by:
\[\sum_{i = 1}^ n r_ix_i = x_1 +_{r_1} \sum_{i = 2}^n \frac{r_i}{1 - r_1}x_i\]
for $n > 2$, $r_1 \neq 1$, with the other cases being evident.
%With this definition, one can prove all the expected laws; further, 
Note that affine maps preserve such finite convex sums. Finite convex sums have been axiomatised as \emph{convex spaces}  (or \emph{convex algebras})
where the sums are required to obey the  \emph{projection} axioims:
\[\sum_{i=1}^n\delta_{ik}x_i = x_k \]
where $\delta_{ik}$ is the Kronecker symbol,  and the \emph{barycentre} axioms: 
\[ \sum_{i=1}^n r_i\big(\sum_{k=1}^m s_{ik}x_k\big) =
      \sum_{k=1}^m\big(\sum_{i=1}^n r_i s_{ik}\big)x_k \]
       %
%for $\mathbf p=(p_1, \dots,p_n)\in P_n$, and  $\mathbf
%q_i=(q_{i1},\dots,q_{im})\in P_m$, $i=1,\dots,n$.
With the above definition of finite sums, barycentric algebras form convex spaces, and indeed the two categories of algebras are equivalent under this correspondence.

%{\tt say what the operations are}

Probability measures on a measurable space, and subclasses of them, naturally form barycentric algebras under the standard binary pointwise convex combination  operations:
\[(\mu +_r \nu) (B) = r\mu(B) + (1-r)\nu(B)\] 
In particular, for any set $X$, the barycentric algebra $P_{f}(X)$ of probability measures over $X$ with finite support is the free barycentric algebra over $X$, with universal arrow the Dirac delta function (see~\cite{Neu,RS}); if $f:X \rightarrow B$ is a map to a barycentric algebra $B$, then the unique extension of $f:X \rightarrow B$ along $\delta$ is given by:
\[\ov{f}(\mu) = \sum_{s \in \supp(\mu)} \mu(s)f(s) \]
Barycentric algebras originate with the work of M.H.~Stone in~\cite{Sto42}, and convex spaces with that of T.~\v{S}wirszcz in~\cite{Swi74}.  For further bibliographic references and historical discussion, see, e.g.,~\cite{KP}.

Turning to quantitative algebras, we work with the category of metric spaces and nonexpansive maps, and its subcategory of complete metric spaces.
%We consider metric spaces as a category, taking metric spaces (possibly empty) as objects and nonexpansive functions as morphisms. 
 These categories have all finite products with the one-point metric space as the final object and with the max metric on binary products $X \times Y$, where:
 \[d_{X \times Y}(\tuple{x,y},\tuple{x',y'}) = \max \{d_X(x,x'),d_Y(y,y')\}\]
We remark that nonexpansive maps are continuous.

A \emph{finitary signature} $\Sigma$ is a collection of operation symbols $f$  and an assignment of an arity $n \in \mathbb{N}$ to each; given such a signature, we write $f:n$ to indicate that $f$ is an operation symbol of arity $n$. A \emph{(metric space) quantitative $\Sigma$-algebra $(X, f_X \; (f \in \Sigma))$} is then a metric space $X$ equipped with a nonexpansive function $f_X: X^n \rightarrow X$ for each operation symbol $f:n$. We often omit the suffix on the operation symbol and also confuse the metric space with the algebra.  A \emph{homomorphism} $h: X \rightarrow Y$ of $\Sigma$-algebras is a nonexpansive map $h: X \rightarrow Y$ such that for all $f:n$ and $x_1,\ldots,x_n \in X$ we have:
\[h(f_X(x_1,\ldots,x_n)) = f_Y(h(x_1),\ldots, h(x_n))\]
This defines a category of quantitative $\Sigma$-algebras and homomorphisms; it has an evident subcategory of complete quantitative $\Sigma$-algebras.

%\section{Wasserstein algebras}

%Wasserstein algebras have a probabilistic choice operation over metric spaces; barycentric algebras form a basic class of algebras with such an operation. 

For any $p \geq 1$, a \emph{Wasserstein (barycentric) algebra of order $p$}  is  a quantitative algebra $(X, +_r: X^2 \rightarrow X \, (r \in [0,1]))$
%metric space $(X,d)$ equipped with a family of non-expansive  binary maps $+_r: X^2 \rightarrow X$ ($r \in [0,1]$) 
forming a barycentric algebra such that for all $x,x',y,y' \in X$ we have:

\begin{equation}d(x +_r y, x' +_r y')^p   \leq  rd(x,x')^p  + (1-r)d(y,y')^p   \label{Was-alg} \tag{$\ast$} \end{equation}

                                                             We remark that the hypothesis of non-expansiveness of the $+_r$  is redundant as, setting \[m = \max\{d(x,x'),d(y,y')\}\] we have:
\[\begin{array}{lcl}d(x +_r y, x' +_r y') 
                                                             & \leq & (rd(x,x')^p  + (1-r)d(y,y')^p)^{1/p} \\
                                                             & \leq & (rm^p  + (1-r)m^p)^{1/p} \\
                                                             & = & m 
\end{array}\]

We gather some other remarks about convex combinations. Recollect that a function $f: X \rightarrow Y$ between metric spaces is \emph{$\alpha$-H\"{o}lder continuous} if, for all $x,y \in X$, we have 
$d(fx,fy) \leq Md(x,y)^{\alpha}$ for some constant $M$.

\begin{lem} \label{was} Let $(X, +_r \, (r \in [0,1]))$ be a Wasserstein algebra of order $p$. Then:
\begin{enumerate}
\item \label{was1} The functions $+_r$ are $r^{1/p}$-Lipschitz in their first argument and $(1 - r)^{1/p}$-Lipschitz in their second argument.
\item \label{was2} Considered as a function of $r$, $x+_r y$ is %$d(x,y)$-Lipschitz
$1/p$-H\"{o}lder continuous.
\item \label{was3} The following  generalisation of equation~(\ref{Was-alg})  to finite convex sums holds:
\[\begin{array}{lclr}d(\sum_{i= 1}^n r_ix_i, \sum_{i=1}^n r_ix'_i)^p 
                                                             & \leq & \sum_{i= 1}^n r_i d(x_i,x'_i)^p 
\end{array}\]
\end{enumerate}
\end{lem}
\begin{proof}
\begin{enumerate}
\item Taking $y = y'$ ($x = x')$ in (\ref{Was-alg}) we respectively obtain:  
\[d(x +_r y, x' +_r y)  \leq r^{1/p}d(x,x') \quad \mbox{and} \quad d(x +_r y, x +_r y')  \leq (1-r)^{1/p}d(y,y')\]
as required.
\item Fix $x,y \in X$ and choose $r,s \in [0,1]$. Suppose that $s \leq r$ and set $e = r - s$.
Assuming $e \neq 1$, we have
\[x +_r y = x +_e (x +_{s/1-e} y) \mbox{   and   } x +_s y = y +_e (x +_{s/1-e} y)\]
So by part 1 we have
\[d(x +_r y,x +_s y) \leq e^{1/p}d(x,y)\]
This also holds when $e = 1$ as then $r = 1$ and $s =0$.
As $e = d(r,s)$, we therefore have:
\[d(x +_r y,x +_s y) \leq d(x,y)d(r,s)^{1/p}\]
By symmetry this also holds when $r \leq s$. As $p \geq1$, $0 < 1/p \leq 1$, so, as a function of $r$, $x +_r y$ is $1/p$-H\"{o}lder continuous, as required.
\item This is a straightforward induction. \qedhere
\end{enumerate}
\end{proof}

%Then for $\delta \in [0,1]$ we have:
%
%\[d(x, y +_{\delta} x) = d(x +_{\delta} x, y +_{\delta} x) \leq  \delta^{1/p}d(x,x'))\]

%Then, for $0 \leq r \leq s \leq 1$, setting  we have:
%\[d(x +_r y, x +_s y) = d(x +_r y, x +_r (x +_{s-r} y))\]

%{\tt move this basic definition to the beginning}

For any metric space $X$, we turn $P_p(X)$ into a barycentric algebra by equipping it with the standard convex combination operations.
% $\mu +_r \nu$, for $r\in [0,1]$, where:
%
%\[(\mu +_r \nu)(B) = r\mu(B) + (1-r)\nu(B)\]
%
It is evident that $\mu +_r \nu$  has  finite  moment of order $p$  if $\mu$ and $\nu$ do.

\begin{lem} \label{lem:Wasalg} For any metric space $X$, $P_p(X)$  
%equipped with the standard convex combination operations $\mu +_r \nu$, for $r\in [0,1]$ 
forms a Wasserstein algebra of order $p$.
\end{lem}
\begin{proof}
%The set  of Radon probability measures on $X$ with finite   moments of order $p$ is evidently closed under the standard convex combination 
%operations $\mu +_r \nu$, for $r\in [0,1]$, and so . 

%It remains
We need only show that for any $\mu, \mu', \nu,\nu' \in P_p(X)$ and $r\in [0,1]$ we have:
\[\begin{array}{lcl}\W_p(\mu +_r \nu, \mu' +_r \nu')^p 
                                                             & \leq & r\W_p(\mu,\mu')^p  + (1-r)\W_p(\nu,\nu')^p 
\end{array}\]
(recalling that theWasserstein condition implies nonexpansivness of the operations).
To this end, choose $\mu, \mu', \nu,\nu' \in P_p(X)$. Let $\alpha$ be a %Radon 
coupling between $\mu$ and $\mu'$ and let $\beta$ be a %Radon 
coupling between $\nu$ and $\nu'$. Then  $\alpha +_r \beta$ is a %Radon 
coupling between $\mu +_r \nu$ and $\mu' +_r \nu'$, and we have:
\[\begin{array}{lcl }\int d^p_X \myd (\alpha +_r \beta) 
                                                                                                      & = &     \  \int d^p_X \myd \alpha +_r \int d^p_X \myd \beta\\
                                                                                                   
\end{array}\]
The conclusion then follows, as we have:
\[\begin{array}[b]{lcl}
\W_p(\mu +_r \nu,\mu' +_r \nu')^p & = &  \inf_{\gamma} \left  (\int d^p_X \myd\gamma \right)\\
                       & \leq &  \inf_{\alpha,\beta} \left  (\int d^p_X \myd(\alpha +_r \beta) \right )\\
                                              & = &  \inf_{\alpha,\beta} \left  (\int d^p_X \myd \alpha +_r \int d^p_X \myd \beta \right )\\
                                              & = &  \inf_{\alpha} \left  (\int d^p_X \myd \alpha \right) +_r  \inf_{\beta} \left  ( \int d^p_X \myd \beta \right )\\
                                              & = & \W_p(\mu,\nu)^p +_r \W_p(\mu',\nu')^p
\end{array} \qedhere\]
%So
%%
%\[\begin{array}{lcl}\W_p(\mu +_r \nu, \mu' +_r \nu')^p 
%                                                             & \leq & r\W_p(\mu,\mu')^p  + (1-r)\W_p(\nu,\nu')^p 
%\end{array}\]
%%
%as required
\end{proof}

%It follows from Lemmas~\ref{was} and~\ref{Wasalg} that Theorem~\ref{Bolgen}  implies the theorem 
%in~\cite{ Bol}, to the effect that  that the metric space of all probability measures with finite  moments 
%of order $p$ over a complete and separable space $X$ is also complete and separable, being 
%generated by the rational measures with finite support in a countable basis of  $X$. For, given the 
%joint nonexpansiveness of $x +_r y$ as a function of $x$ and $y$ and its H\"{o}lder continuity as a 
%function of $r$,  it follows that every finitely supported probability measure is a limit of rational convex 
%combinations of finitely supported probability measures with support included in a generating set of 
%$X$.

We next turn to characterising the free Wasserstein algebras of order $p$. For any metric space X, let $F_p(X)$ be the sub-Wasserstein algebra  of order $p$ of $P_p$ that consists of the finitely supported probability measures.
\begin{thm} \label{atfree} For any metric space X, $F_p(X)$ is the free Wasserstein algebra of order $p$ over $X$, with universal arrow the Dirac delta function $\delta:X \rightarrow F_p(X)$.
\end{thm} 
\begin{proof} 
First,  $\delta:X \rightarrow F_p(X)$ is an isometric embedding, since, for any $x,y \in X$, the only %(and evidently Radon)  
coupling between $\delta(x)$ and $\delta(y)$ is $\delta(\tuple{x,y})$, and so we have:
\[W_p(\delta(x)), \delta(y)) = (d(x,y)^p)^{1/p} = d(x,y)\]

Next, let $f: X \rightarrow W$ be any nonexpansive affine map to a Wasserstein algebra $B$ of order $p$. 
As a set, $F_p(X) $ is the free barycentric algebra over $X$ with unit the Dirac function, 
and the unique affine map $\ov{f}: F_p(X) \rightarrow W$ extending $f$ along the unit is given by the formula:  
\[\ov{f}(\mu) = \sum_{s \in \supp(\mu)} \mu(s)f(s) \]
So it only remains to show that $\ov{f}$ is nonexpansive. Let $\mu$, $\nu$ be finitely supported probability measures with respective finite non-empty supports $S$ and $T$. Let $\gamma$ be a %Radon 
coupling between them. Then $\gamma$ has support $\subseteq S \times T$.  Consequently we can write it as a convex sum, as follows:
\[\gamma = \sum_{s \in S, t \in T} \gamma(\tuple{s,t})\delta(\tuple{s,t})\]
As $\mu$ and $\nu$ are the marginals of $\gamma$, we have:
\[\mu = \sum_{s \in S, t \in T} \gamma(\tuple{s,t})\delta(s) \quad \mbox{and}\quad \nu = \sum_{s \in S, t \in T} \gamma(\tuple{s,t})\delta(t)\]
We can now calculate:
\[\begin{array}{lcl}
 d(\ov{f}(\mu),\ov{f}(\nu)) & = & d(\ov{f}(\sum_{s \in S, t \in T} \gamma(\tuple{s,t})\delta(s)),\ov{f}(\sum_{s \in S, t \in T} \gamma(\tuple{s,t})\delta(t)))\\
                                        & = & d(\sum_{s \in S, t \in T} \gamma(\tuple{s,t})\ov{f}(\delta(s)),\sum_{s \in S, t \in T} \gamma(\tuple{s,t})\ov{f}(\delta(t)))\\
                                        & = & d(\sum_{s \in S, t \in T} \gamma(
                                        \tuple{s,t})f(s),(\sum_{s \in S, t \in T} \gamma(\tuple{s,t})f(t))\\
                                                                                &&  \hspace{165pt} (\mbox{as $f = \ov{f}\delta$})\\
                                        & \leq  & \left ( \sum_{s \in S, t \in T} \gamma(\tuple{s,t})d(f(s),f(t))^p \right )^{1/p} \\
                                        &&  \hspace{165pt}  (\mbox{by Lemma~\tref{was}{was3}})\\
                                        & \leq  &  \left (\sum_{s \in S, t \in T} \gamma(\tuple{s,t})d(s,t)^p \right )^{1/p} \\
                                                                                && \hspace{165pt} (\mbox{as $f$ is nonexpansive})\\
                                       & = &  \left ( \int d^p_X \myd\gamma \right )^{1/p} 
                                        \end{array}\]

As $\gamma$ was an arbitrary %Radon 
coupling between $\mu$ and $\nu$ we therefore have
\[ d(\ov{f}(\mu),\ov{f}(\nu)) \leq W_p(\mu,\nu)\]
as required.
\end{proof}

To characterise the free complete Wasserstein algebras of order $p$, we need two lemmas, one on the completion of metric spaces, and the other on the completion of quantitative $\Sigma$-algebras.
The usual completion $\ov{X}$ of a metric space $X$ by equivalence classes $[\bx]$ of Cauchy sequences is the free complete metric space over $X$, with universal arrow the isometric embedding $c_X: X \rightarrow \ov{X}$, where $c_X(x) = [\tuple{x}_{n \in \mathbb{N}}]$. That is, for any complete metric space $Y$ and nonexpansive function $f: X \rightarrow Y$ there is a unique nonexpansive function $\ov{f}: \ov{X} \rightarrow Y$ extending $f$ along $c_X$. This function is given by the formula:
\[\ov{f}([\mathbf{x}]) = \lim_n f(\mathbf{x}_n)\]
There is a simple criterion for when an isometric embedding of a metric space in another is a universal arrow:
\begin{lem} \label{gen}Let $\theta: X \rightarrow Y$ be an isometric embedding of  metric spaces with $Y$ complete, and $\theta(X)$ generating $Y$. Then $Y$ is the free complete metric space over $X$, with universal arrow  $\theta$.
\end{lem}
\begin{proof}
We can assume without loss of generality that $\theta$ is an inclusion mapping. 
%Let $c: X \rightarrow \ov{X}$ be the usual universal metric completion of $X$ by Cauchy sequences. 
%The map $\ov{\theta}: \ov{X} \rightarrow Y$ extends the inclusion $\theta$  along $c_X$.
We show that  $\ov{\theta}$ is an isometry. To show it is onto, choose $y \in Y$. As $X$ generates $Y$,  $y$ is a limit of some Cauchy sequence $\bx$ in $X$ and we have
\[\ov{\theta}([\bx])  = \lim_n \theta(\bx_n) = \lim_n \bx_n = y\]
as required.

To show $\ov{\theta}$ preserves distances, we calculate:
\[d(\ov{\theta}([\mathbf{x}]),\ov{\theta}([\mathbf{y}]) = d( \lim_n \mathbf{x}_n, \lim_n \mathbf{y}_n ) =  \lim_n d(\mathbf{x}_n,  \mathbf{y}_n ) = d([\bx],[\mathbf{y}])\]
So $\ov{\theta}$ is indeed an isometry. It follows that $\theta$ is  a universal arrow as $c_X$ is, and we are done.
\end{proof}

The criterion of Lemma~\ref{gen} extends to $\Sigma$-algebras:
\begin{lem} \label{alggen} Let $X$, $Y$ be quantitative $\Sigma$-algebras and let $\theta: X \rightarrow Y$ be both an algebra homomorphism and an isometric embedding. Suppose that $Y$ is complete as a metric space, and $\theta(X)$ generates $Y$ as a metric space. Then $Y$ is the free complete quantitative $\Sigma$-algebra over $X$, with universal arrow $\theta$.
\end{lem} 
\begin{proof}
We can assume without loss of generality that $\theta$ is an inclusion mapping. Let  $h:X \rightarrow Z$ be any nonexpansive homomorphism from $X$ to a complete quantitative algebra $Z$. By Lemma 1, $h$ extends to a unique nonexpansive map $\ov{h}: Y \rightarrow Z$, and we only have to show that this is a homomorphism. Taking a binary operation $f$ as an example, let $x,y$ be in $Y$. As $X$ generates $Y$ as a metric space there are Cauchy sequences $\bx$ and $\bby$ in $X$ with respective limits $x$ and $y$, and we calculate:
\[\begin{array}[b]{lcl}\ov{h}(f_Y(x,y)) & = & \ov{h}(f_Y(\lim_n \bx_n,\lim_n \bby_n)) \\
                                                     & = &  \ov{h}(\lim_n f_Y( \bx_n, \bby_n)) \\
                                                     & = &  \lim_n \ov{h}( f_Y( \bx_n, \bby_n))\\
                                                     & = &  \lim_n \ov{h}( f_X( \bx_n, \bby_n))\\
                                                     & = &  \lim_n h( f_X( \bx_n, \bby_n))\\
                                                     & = &  \lim_n  f_Z(h( \bx_n), h(\bby_n))\\
                                                     & = &    f_Z(\lim_n h( \bx_n), \lim_n h(\bby_n))\\
                                                     & = &    f_Z(\lim_n \ov{h}( \bx_n), \lim_n \ov{h}(\bby_n))\\
                                                     & = &    f_Z( \ov{h}(\lim_n \bx_n),  \ov{h}(\lim_n \bby_n))\\
                                                     & = &    f_Z( \ov{h}(x),  \ov{h}(y))\\
\end{array} \qedhere \]
\end{proof}

We can now prove our main theorem on free complete Wasserstein algebras of order $p$:
\begin{thm}  \label{wholepoint} For any complete metric space X, $P_p(X)$ is the free complete Wasserstein algebra of order $p$ over $X$, with universal arrow the Dirac delta function.
\end{thm}

\begin{proof} Theorem~\ref{Bolgen} and Lemma~\ref{lem:Wasalg} tell us that $P_p(X)$ is a complete Wasserstein algebra of order $p$.  It therefore remains to prove it is the free one over $X$, with unit the Dirac delta function.

First, by Theorem~\ref{atfree},  $F_p(X)$ is the free Wasserstein algebra of order $p$ over $X$, with universal arrow the Dirac function.  

Second, $P_p(X)$ is the free complete Wasserstein algebra of order $p$  over $F_p(X)$, with universal arrow the inclusion.  
To see this, first note that  $F_p(X)$ is  a sub-quantitative algebra of $P_p(X)$ that, by  Theorem~\ref{Bolgen}, generates $P_p(X)$,
So, by Lemma~\ref{alggen}, $P_p(X)$ is the free complete quantitative algebra with signature the $+_r$ over $F_p(X)$, with universal arrow the inclusion.
 It is therefore, in particular, the free complete Wasserstein algebra of order $p$  over $F_p(X)$, with universal arrow the inclusion.

%Further, from Theorem~\ref{Bolgen} we also have that $P_p(X)$ is generated by $F_p(X)$, and that is both a subalgebra and a submetric space of $P_p(X)$. 
%By Lemma~\ref{alggen} we then have that $P_p(X)$ is the free complete quantitative algebra with signature the $+_r$ over $F_p(X)$, with universal arrow the inclusion.
%It is therefore, in particular, the free complete Wasserstein algebra of order $p$  over $F_p(X)$, with universal arrow the inclusion.

Combining these two free algebra assertions, we see that,  as required,  $P_p(X)$ is the free complete Wasserstein algebra  of order $p$ over $X$, with universal arrow the Dirac delta function.
\end{proof}

There is a disconnect between the present work and that in~\cite{MPP} where extended metrics are used. We sketch the  straightforward extension of the above freeness results to extended metric spaces. An extended metric is a  function $d:X^2 \rightarrow \eR$ to the extended reals, defined in the usual way, but using the natural extension of addition to $\eR$, where $x + \infty = \infty +x = \infty$ and the natural extension of the order where $x \leq \infty$; see~\cite{BBI} for information on such spaces. 
The topology on an extended metric space $(X,d)$ is defined in the usual way using open balls. Nonexpansive functions are defined in the usual way, and are continuous. Cauchy sequences are defined in the usual way, as are then complete extended metric spaces. The completion of an extended metric space is defined as usual and the function $c_X: X \rightarrow \ov{X}$ so obtained is an isometry and is universal. 

The analogue of Lemma~\ref{gen} then goes through, with the same proof, as does that of Lemma~\ref{alggen}, with the evident definitions of quantitative $\Sigma$-algebras over extended metric spaces and their homomorphisms. Next, the Wasserstein algebras of order $p$ over an extended metric space are also defined as before, and
the extension of Lemma~\ref{lem:Wasalg} to extended metric spaces then goes through, with the same proofs, as do the following extensions of Theorems~\ref{atfree} and ~\ref{wholepoint}:
\begin{thm} \label{atfreeex} For any metric space X, $F_p(X)$ is the free extended metric Wasserstein algebra of order $p$ over $X$, with universal arrow the Dirac delta function $\delta:X \rightarrow F_p(X)$.
\end{thm} 

\begin{thm}  For any complete metric space X, $P_p(X)$ is the free complete extended metric Wasserstein algebra of order $p$ over $X$, with universal arrow the Dirac delta function.
\end{thm}

We remark that in~\cite{MPP} the Wasserstein algebras of order $p$ are axiomatised using quantitative equational logic. This logic makes use of quantitative equations, which have the form $t =_r u$, where $t$ and $u$ are terms built as usual from a finitary signature, and $r\in [0,1]$. If $t$ and $u$ denote elements $a$ and $b$ of an extended metric space $X$, then  $t =_r u$ holds if, and only if, $d_X(a,b) \leq r$.
The  following axiom scheme  of quantitative equational logic  axiomatises the Wasserstein algebras of order $p$:
%
%\[\{x=_{q_1} y, x'=_{q_2}y'\}\vdash x+_r x'=_{e} y+_r y'\]
%
\[x=_{q_1} y, x'=_{q_2}y'\vdash x+_r x'=_{e} y+_r y'\]
where $q_1,q_2,q,e$ range over rationals  in $[0,1]$  such that $rq_1^p+(1-r)q_2^{p} \leq e^p$.

Our discussion has been in terms of barycentric algebras, but, as common in algebra, other axiomatisations are possible. As an immediate example, noting that the categories of barycentric algebras and convex spaces are equivalent, let us define Wasserstein convex spaces of order $p$ to be quantitative convex spaces which obey the natural Wasserstein condition of order $p$ for convex spaces, viz that, for all $n \geq 1$ we have:
\[\begin{array}{lclr}d(\sum^n_{i = 1} r_ix_i, \sum^n_{i = 1} r_ix'_i)^p 
                                                             & \leq & \sum^n_{i = 1}  r_i d(x_i,x'_i)^p 
\end{array}\]
Then  the categories of Wasserstein barycentric algebras of order of $p$ and  Wasserstein convex spaces of order of $p$ are equivalent, as are their subcategories of algebras over complete metric spaces. One then immediately obtains from the above work characterising  free Wasserstein barycentric algebras of order of $p$  corresponding results characterising free Wasserstein convex spaces of order of $p$. 

The situation is more complex if we consider instead midpoint (or mean-value or medial mean) algebras, see, e.g.,~\cite{Ker,Hec}. 
These algebras are the commutative, idempotent, medial groupoids. That is, they have a binary operation, which we write as $\oplus$, subject to the following laws:
\begin{align} 
x \oplus y& = y \oplus x \tag{C}\\  
x \oplus x& = x \tag{I}\\
(x \oplus u) \oplus (v \oplus z) & = (x \oplus v) \oplus (u \oplus z) \tag{M}
\end{align} 
One can define binary convex combinations $x +_r y$ in these algebras where $r \leq 1$ is a dyadic fraction (i.e., one  of 
the form $\frac{m}{2^n}$), and they obey analogues of the barycentric axioms. The free such algebra over a set consists of all finite probability distributions with dyadic weights, see~\cite[Theorems 2.14 and 2.15]{JK}.

Every barycentric algebra evidently yields a midpoint algebra by setting $x \oplus y = x +_{0.5} y$, and we then have a full and faithful functor to the category of midpoint algebras,  which acts as the identity on morphisms; however it is not an equivalence of categories, as not every midpoint algebra can be extended to a barycentric algebra with the same carrier
\footnote{The free midpoint algebra over two generators ($a$ and $b$, say) cannot be so extended. It consists of all formal dyadic fractional  convex combinations of $a$ and $b$. If, for the sake of contradiction, there was such an extension then we would have $a +_r b = a +_s b$ for some $r$ and $s$ with $r$ not a dyadic fraction and $s$ a dyadic fraction. But,  according to a lemma of Neumann (see~\cite[Lemma 4]{Neu}), if the equation $a +_r b = a +_s b$ holds for two different $r, s \in (0,1)$ in a barycentric algebra, then it holds for all different 
$r, s \in (0,1)$ in that algebra. So, in our case, taking $r, s \in (0,1)$ to be two different dyadic fractions we obtain the desired contradiction.}.

The situation changes when we move to quantitative algebras. Define 
Wasserstein midpoint algebras of order $p$ to be quantitative midpoint algebras which satisfy the condition
\begin{align}d(x \oplus y, x' \oplus y')^p   &\leq  d(x,x')^p  \oplus d(y,y')^p   \tag*{} \end{align}
Then,  following the argument of Theorem~\ref{atfree}, we see that, for any metric space X, $F_p(X)$, with its natural midpoint algebra structure, is the free Wasserstein midpoint algebra of order $p$ over $X$. Next, note that, for complete $X$,  $P_p(X)$ is even generated by the finite probability distributions with dyadic fractional weights, as, by Lemma, $x +_r y$ is H\"{o}lder continuous in $r$, and so, if a sequence $r_n$ in $[0,1]$ converges to $r$, then the sequence $x +_{r_n} y$ is Cauchy and converges to $x +_r y$. With that, now following the argument of Theorem~\ref{wholepoint}, we see that $P_p(X)$ forms the free complete Wasserstein midpoint algebra of order $p$ over $X$. 

The construction of midpoint algebras from barycentric algebras yields a full and faithful functor from the category of Wasserstein barycentric algebras of order $p$ over complete metric spaces to Wasserstein midpoint algebras of order $p$ over complete metric spaces. We now sketch a proof  that this is in fact an equivalence of categories.
It suffices to show that every complete Wasserstein midpoint algebra
of order $p$ can be extended to a complete Wasserstein barycentric
algebra of order $p$ with the same carrier. Let $A$ be such a midpoint
algebra.
There is a unique nonexpansive morphism $h: P_p(A) \rightarrow A$ of Wasserstein midpoint algebras of order $p$ extending the identity on $A$ along the unit. Using this, one inherits the H\"{o}lder continuity of $+_r$ on $A$, for dyadic $r$ from that on $P_p(A)$. Then one can define $x +_r y \in A$ for $x,y \in A$ as the limit of the sequence $x +_{r_n}  y$ where $r_n$ is a chosen sequence of dyadic fractions converging to $r$, noting that the sequence $x +_{r_n}  y$ is Cauchy by H\"{o}lder continuity.

%Curiously, the first algebraic characterisation result on the action of the Kantorovich monad was in terms of midpoint algebras. It was essentially shown in 
%~\cite{BHM,BHMad} that, restricted to 1-bounded separable complete metric spaces, the monad gives the free Kantorovich midpoint algebra (meaning the 
%free Wasserstein midpoint algebra of order 1).

%{\tt figure out the exact thing to say about ~\cite{BHM}}


\begin{thebibliography}{1}

\bibitem{Bas} G.~Basso, 
A Hitchhiker's guide to Wasserstein distances, available at: \url{
http://n.ethz.ch/~gbasso/}, 2015.

\bibitem{VanB} F.~van Breugel, The metric monad for probabilistic
nondeterminism, 
Unpublished note, available at
\url{http://www.cse.yorku.ca/~franck/research/drafts/monad.pdf},
2005.

\bibitem{BBI}  D.~Burago, Y.D.~Burago, and S.~Ivanov, \emph{A course
in metric geometry}, Vol.\ 33, American Mathematical Society, 2001.

  \bibitem{CD} P.~Clement  and  W.~Desch,
  An elementary proof of the triangle inequality for the Wasserstein
metric, \emph{Proceedings of the American Mathematical Society}, {\bf
136}(1), 333--339, 2008.

\bibitem{Edw} D.E.~Edwards,
A simple proof in Monge-Kantorovich duality theory, \emph{Studia
Math.}, {\bf 200}(1), 67--77, 2010.

\bibitem{FP} T.~Fritz and P.~Perrone, 
A Probability Monad as the Colimit of Finite Powers, arXiv:
1712.05363, 2017.

\bibitem{Fre} D.H.~Fremlin, \emph{ Measure Theory, Volume 4:
Topological Measure Spaces}, Torres Fremlin, 2006.

\bibitem{Gir} M.~Giry,  A categorical approach to probability theory,
\emph{Categorical aspects of topology and analysis},  Lecture Notes
in Mathematics, Vol. 915, 68--85, Springer, 1982.

\bibitem{Hec} R.~Heckmann, Probabilistic domains, \emph{Proc.\  19th.
\ International Colloquium on Trees in Algebra and
Programming}, Lecture Notes in Computer Science, Vol.\ 787, 142--156,
Springer, 1994.

\bibitem{HKSY} C.~Heunen, O.~Kammar, S.~Staton, and H.~Yang,
A convenient category for higher-order probability theory,
\emph{Proc.\  32nd.\ Annual Symposium on Logic in Computer Science},
1--12, IEEE, 2017.

\bibitem{JK} J.~Je\v{z}ek and T.~Kepka,  Free commutative idempotent
abelian groupoids and quasigroups, \emph{Acta Universitatis
Carolinae, Mathematica et Physica},  {\bf 17}(2), 13--19, 1976.

\bibitem{Jon} C.~Jones, \emph{Probabilistic non-determinism}, Ph.D.\
Thesis, Department of Computer Science, The University of Edinburgh,
197 pages,
available at \url{ https://www.era.lib.ed.ac.uk/handle/1842/413},
1990.

\bibitem{JP}  C.~Jones and G.D.~Plotkin,   A probabilistic
powerdomain of evaluations, \emph{Proc.\ 4th.\ Annual Symposium on 
Logic in Computer Science}, pp. 186--195, IEEE, 1989.

\bibitem{KP} K.~Keimel and G.D.~Plotkin, 
Mixed powerdomains for probability and nondeterminism, \emph{Logical
Methods in Computer Science}, {\bf 13}(1:2), 1--84, 2017.

\bibitem{Ker} S.~Kermit, Cancellative medial means are arithmetic,
\emph{Duke Math J.}, {\bf 37}, 439--445,1970.

\bibitem{Law} F.W.~Lawvere, The category of probabilistic mappings,
Preprint, available at \url{https://ncatlab.org/nlab/files/
lawvereprobability1962.pdf}, 1962.

\bibitem{MPP} R.~Mardare, P.~Panangaden, and G.D.~Plotkin,
Quantitative Algebraic Reasoning, \emph{Proc.\ 31st Annual Symposium
on Logic in Computer Science}, 700--709, 2016.

\bibitem{Neu} W.D.~Neumann, 
\newblock On the quasivariety of convex subsets of affine spaces, 
\newblock {\em Archiv der Mathematik}, 21:11--16, 1970.

\bibitem{Par} K.R.~Parthasarathy, \emph{Probability measures on
metric spaces}, Academic Press, 1967.

\bibitem{RS} A.B.~Romanowska and J.D.H.~Smith, \emph{Modes}, World
Scientific, Singapore, 2002.

\bibitem{Sto42}  M.H.~Stone, Postulates for the barycentric calculus,
\emph{Annali di Matematica Pura ed Applicada}, Serie IV,
{\bf 29}, 25--30, 1942.

\bibitem{Swi74} T.~\v{S}wirszcz, Monadic functors and convexity, \
emph{Bulletin de l'Acad\'{e}mie Polonaise des Sciences},
Ser.\ Sci.\ Math.\ Astronom.\ Phys, {\bf 22}, 39--42, 1974.

\bibitem{Vil03} C.~Villani, \emph{Topics in optimal transportation},
Vol.\ 58 of Graduate Studies in Mathematics, American Mathematical
Society, 2003.

\bibitem{Vil08} C.~Villani,
\emph{Optimal transport: old and new}, Vol.\ 338, Springer Science \&
Business Media, 2008.

\end{thebibliography}
\end{document}